\documentclass[letterpaper, 11pt]{article}
\usepackage{times}
\usepackage{graphicx}
\usepackage{color}
\usepackage{mathrsfs}
\usepackage{amsthm}
\usepackage{latexsym}
\usepackage{amssymb}
\usepackage{amsmath}
\usepackage[right=0.8in, top=1in, bottom=1in, left=0.8in]{geometry}
\usepackage{setspace}
\spacing{1.1}

\newcommand{\nc}{\newcommand}  
\nc{\Pf}{{\rm Pf}} \nc{\PfS}{{\rm PfS}} \def\uu{\mathbf{u}}
\nc{\GP}{Grassmann-Pl{\"u}cker } \nc{\adj}{{\rm adj}}
\def\AA{\mathbf{A}}  
  
  \def\00{\mathbf{0}}
 \def\qq{\mathbf{q}} 
 \def\xx{\mathbf{x}} 
  \def\yy{\mathbf{y}}

 \def\11{\mathbf{1}} 
  \def\vv{\mathbf{v}}
\def\ww{\mathbf{w}}  
  
\def\calF{\mathcal{F}}  
  
\def\ww{\mathbf{w}} \def\zz{\mathbf{z}}

   \def\tt{\mathbf{t}}

   \def\MM{\mathbf{M}}
\def\ss{\mathbf{s}} 

\newtheorem{coro}{Corollary} \newtheorem{defi}{Definition} 
\newtheorem{lemm}{Lemma}  \newtheorem{theo}{Theorem}

\begin{document}

\title{\vspace{1.5cm}Non-negative Weighted \#CSPs: An Effective Complexity Dichotomy\vspace{1.2cm}}
\author{Jin-Yi Cai\\ University of Wisconsin, Madison \and
Xi Chen\\ Columbia University \and Pinyan Lu\\ Microsoft Research Asia\vspace{0.8cm} }
\date{}
\maketitle

\begin{abstract}
We prove a complexity dichotomy theorem for all
non-negative weighted counting Constraint Satisfaction Problems (CSP).
This caps a long series of important results on counting problems
including unweighted and weighted graph homomorphisms~\cite{DyerGreenhill,
BulatovGrohe,DirectedHomo,CaiChen} and the celebrated dichotomy theorem for
unweighted \#CSP~\cite{BulatovDalmau,Bulatov,Dyer-Rich,Dyer-Rich2}.
Our dichotomy theorem gives a succinct criterion for tractability.
If a set $\cal F$ of constraint functions satisfies the criterion, then the
counting CSP problem defined by $\cal F$ is solvable in polynomial time;
if it does not satisfy the criterion, then the problem is \#P-hard.
We furthermore show that the question of whether $\cal F$
satisfies the criterion is decidable in NP.

Surprisingly, our tractability criterion is simpler than the previous
criteria for
the more restricted classes of problems, although when specialized to
those cases, they are logically equivalent.
Our proof mainly  uses Linear Algebra, and represents a departure
from Universal Algebra, the dominant methodology in recent years.

\end{abstract}
\thispagestyle{empty}
\newpage

\setcounter{page}{1}

\section{Introduction}

The study of Constraint Satisfaction Problems (CSP) has been one of the most
  active research areas, where enormous progress has been made in recent years.
The investigation of CSP includes at least the following major~bran\-ches:
  Decision Problems --- whether a solution exists \cite{Schaefer,Hell,BulatovDecision,SzegedyKun};
  Optimization Problems --- finding a solution that satisfies the most constraints
  (or in the weighted case achieving the highest total
weight)~\cite{Hastad01,KhotKMO07,AustrinM08,DinurMR09,Raghavendra,Tulsiani,RaghavendraSteurer};
  and Counting Problems --- to count the number of solutions, including its weighted
version~\cite{BulatovDalmau,Bulatov,BulatovSurvey,Rational,Dyer-Rich}.
The decision CSP dichotomy conjecture of Feder and Vardi~\cite{Feder-Vardi}, that
  every decision CSP problem defined by a constraint language $\Gamma$ is either in P or NP-complete, remains open.
A great deal of work has been devoted to the optimization version of CSP, constituting a significant
  fraction of on-going activities in approximation algorithms.

The subject of this paper is on counting CSP; more precisely on
  \emph{weighted} counting Constraint Satisfaction~Pro\-blems, denoted as weighted \#CSP.
For {\it unweighted}   \#CSP, the
problem is usually stated as follows: $D$ is a fixed finite set called
the domain set. A fixed finite set
of constraint predicates $\Gamma = \{\Theta_1, \ldots, \Theta_h\}$
is given,
where each $\Theta_i$ is a relation on $D^{r_i}$ of some finite arity $r_i$.
Then an instance of \#CSP$(\Gamma)$ consists of a finite set of
variables $x_1, \ldots, x_n$, ranging over $D$, and a finite set of
constraints from $\Gamma$, each applied to a subset of these variables.
It defines a new $n$-ary relation $R$ where $(x_1,\ldots,x_n)\in R$ if and only if
  all the constraints are satisfied.
The \#CSP problem then asks for the size of $R$.
In a (non-negatively) weighted \#CSP,
the set $\Gamma$ is replaced by a fixed finite set of
constraint functions, ${\cal F} = \{f_1, \ldots, f_h\}$, where each
$f_i$ maps $D^{r_i}$ to non-negative reals  $\mathbb{R}_{+}$.
An instance of \#CSP$({\cal F})$ consists of
variables $x_1, \ldots, x_n$, ranging over $D$, and a finite set of
constraint functions from ${\cal F}$,
each  applied to a subset of these variables.
It defines a new $n$-ary function $F$: for any assignment $(x_1,\ldots,x_n)$,
  $F(x_1,\ldots,x_n)$ is the product of the constraint function evaluations.
The output is then the so-called {\it partition function}, that is,
the sum of $F$ over all assignments
$\{x_1, \ldots, x_n\} \rightarrow  D$.
The unweighted \#CSP
is the special case where each constraint function
is $0$-$1$ valued.
(A formal definition will be given in Section~\ref{preliminary}.)

Regarding unweighted \#CSP, Bulatov~\cite{Bulatov} proved a sweeping
dichotomy theorem. He gave a criterion, {\it
congruence singularity}, and showed that for any finite set of constraint
predicates $\Gamma$ over any finite domain $D$, if $\Gamma$
satisfies this condition, then \#CSP($\Gamma$) is solvable in P;
otherwise it is \#P-complete.  His proof uses deep structural
theorems from universal algebra \cite{universal,HobbyMcKenzie,FreeseMcKenzie}.  Indeed this approach using
universal algebra has been one of the most exciting developments
in the study of the complexity of CSP in recent years, first used in decision CSP \cite{Jeavons98,JeavonsCohen,BulatovDecision,BulatovLogic}, and has been called the
{\it Algebraic Approach}.

However, this is not the {\it only} approach.
In~\cite{Dyer-Rich}, Dyer and Richerby gave an alternative proof
of the dichotomy theorem for unweighted \#CSP. Their proof
is considerably more direct, and uses no universal algebra other than
the notion of a Mal'tsev polymorphism.
They also showed that the dichotomy is decidable \cite{Dyer-Rich3,Dyer-Rich2}.
Furthermore,
by treating rational weights as integral multiples of
a common denominator, the dichotomy theorem can be extended
to include positive rational weights~\cite{Rational}.

In this paper, we give a complexity dichotomy theorem for all
  non-negative weighted \#CSP($\cal F$).  To describe our approach, let us first
briefly recap the proofs by Bulatov and by Dyer and Richerby.
Bulatov's proof is deeply embedded in a structural theory
of universal algebra
called {\it tame congruence theory} \cite{HobbyMcKenzie}.
(A congruence is an equivalence relation expressible in a given
universal algebra.)
The starting point of this {\it Algebraic Approach}
is the realization of a close connection between
  unweighted \#CSP($\Gamma$) and the {\it relational clone}
$\langle \Gamma \rangle$ generated by $\Gamma$.  $\langle \Gamma \rangle$
is the closure set of all relations expressible
from $\Gamma$ by boolean conjunction $\wedge$
and the existential quantifier $\exists$. A basic property,
called {\it congruence permutability}, is then shown to be a necessary
condition for the tractability of \#CSP($\Gamma$)~\cite{BulatovJeavons, BulatovDalmau, BulatovSurvey}.
It is known from universal algebra that congruence permutability
is equivalent to the existence of {\it Mal'tsev polymorphisms}.
It is also equivalent to the more combinatorial condition of {\it
strong rectangularity} of Dyer and Richerby \cite{Dyer-Rich}:
For any $n$-ary relation $R$ defined by an instance of \#CSP$(\Gamma)$,
  if we partition its $n$ variables into three parts: $\uu=(u_1,\ldots,u_k),\vv=(v_1,\ldots,v_\ell)$
  and $\ww=(w_1,\ldots,w_{n-k-\ell})$, then the following $|D|^k\times |D|^\ell$ matrix $\MM$ must be \emph{block-diagonal}
  after separately permuting its rows and columns:
$M(\uu,\vv)=1$ if there exists a $\ww$ such that $(\uu,\vv,\ww)\in R$; and $M(\uu,\vv)=0$ otherwise.
(See the formal definition in Section \ref{preliminary}.)

Assuming $\Gamma$ satisfies this necessary condition (otherwise \#CSP$(\Gamma)$ is already \#P-hard),
  Bulatov's proof delves much more deeply than Mal'tsev polymorphisms and uses
a lot more results and techniques from universal algebra.
 The Dyer-Richerby proof manages to avoid much of universal algebra.
They went on to give a more combinatorial criterion, called {\it strong balance}:
For any $n$-ary relation $R$ defined by an instance of \#CSP$(\Gamma)$,
  if we partition its $n$ variables into four parts: $\uu=(u_1,\ldots,u_k),\vv=(v_1,\ldots,v_\ell),
  \ww=(w_1,\ldots,w_{t}),\zz=(z_1,\ldots,z_{n-k-\ell-t})$, then the following $|D|^k\times |D|^\ell$ integer
  matrix $\MM$ must be {block-diagonal} and \emph{all of its blocks are of rank $1$} (which we will refer to as a \emph{block-rank-1} matrix):
\begin{equation}\label{exist}
M(\uu,\vv)=\Big|\big\{\ww:\exists\hspace{0.08cm} \zz\ \text{such that}\ (\uu,\vv,\ww,\zz)\in R\big\}\Big|,
\ \ \ \ \ \text{for all $\uu\in D^k$ and $\vv\in D^\ell$.}
\end{equation}
(See the formal definition in Section \ref{sec:equiv}.)
Dyer and Richerby \cite{Dyer-Rich} show that strong balance (which implies strong rectangularity) is the criterion for
  the tractability of \#CSP($\Gamma$).
They further prove that it is equivalent to Bulatov's criterion of
{\it congruence singularity} which is stated in the language of
universal algebra.


The first difficulty we encountered when trying to extend the unweighted dichotomy to weighted \#CSP$(\calF)$~is that
  there is no direct extension of the notion of strong balance above in the weighted world.
While the number of $\ww$ satisfying $R$ on the right side of (\ref{exist})
  can be naturally replaced by the sum of $F$ (any function defined by an \#CSP$(\calF)$ instance) over $\ww$,
  we do not see any easy way to introduce existential quantifiers to this more general weighted setting.
Moreover, the use of existential quantifiers in the notion of strong balance is
  crucial to the proof of Dyer and Richerby:
  their polynomial-time counting algorithm for tractable \#CSP$(\Gamma)$ heavily relies on them.

While there seems to be no natural notion of an existential quantifier
in the weighted setting, we came to~a key observation  that
the notion of strong balance
  is equivalent to the one {\it without using any existential~quan\-tifiers}
 (that is, we only consider
  partitions of the variables into $3$ parts with no $\zz$).
We include the proof of this equivalence in Section \ref{sec:equiv}.
This inspires us to use the following seemingly weaker notion of \emph{balance} for weighted \#CSP$(\calF)$, with no
  existential quantifiers at all:
For any $n$-ary function $F$ defined by a \#CSP$(\calF)$ instance, if we partition its $n$
  variables into three parts: $\uu=(u_1,\ldots,u_k),\vv=(v_1,\ldots,v_\ell)$
  and $\ww=(w_1,\ldots,w_{n-k-\ell})$, then the following $|D|^k\times |D|^\ell$ matrix $\MM$ must be \emph{block-rank-1}:
$$
M(\uu,\vv)=\sum_{\ww\in D^{n-k-\ell}} F(\uu,\vv,\ww),\ \ \ \ \ \text{for all $\uu\in D^k$ and $\vv\in D^\ell$.}
$$

It is easy to show that balance is a necessary condition for the tractability of
  \#CSP$(\calF)$.
But is it also sufficient?
If $\calF$ is balanced, can we solve it in polynomial time?
We show that this is indeed the case by giving a polynomial time
 counting scheme for
  all \#CSP$(\calF)$s with $\calF$ being balanced.
Our algorithm works differently from the one of Dyer and Richerby.
It avoids the use of existential quantifiers and is designed specially for weighted
  and balanced \#CSP$(\calF)$s.
As a result, we get the following dichotomy for non-negatively weighted \#CSP
  with a logically simpler criterion:\vspace{0.03cm}

\begin{theo}[Main]\label{mainmain}
\#CSP$(\calF)$ is in polynomial-time if $\calF$ is balanced; and is \#P-hard otherwise.\vspace{0.03cm}
\end{theo}

A new ingredient of our proof is the concept of \emph{a vector representation}
for a
  non-negative function.
Let $F$ be a function over $x_1,\ldots,x_n$. Then $s_1,\ldots,s_n:D\rightarrow \mathbb{R}_+$
  is a vector representation of $F$ if for any $(x_1,\ldots,x_n)\in D^n$ such that
  $F(x_1,\ldots,x_n)>0$, we have $F(x_1,\ldots,x_n)=s_1(x_1)\cdots s_n(x_n).$
The first step of our algorithm is to show that given any instance of \#CSP$(\calF)$, where
  $\calF$ is balanced, the function it defines has a vector representation which can be
  computed in polynomial time.
However, $F$ may have a lot of ``holes'' where $s_1(x_1)\cdots s_n(x_n)>0$
  but $F(x_1,\ldots,x_n)=0$ so it is still not clear how to do the sum of $F$ over $x_1,\ldots,x_n$.

The next step is quite a surprise.
Assuming $\calF$ is balanced, we show how to construct one-variable functions
  $t_2,\ldots,t_n:D\rightarrow \mathbb{R}_+$ in polynomial time such that
  for any $(u_1,\ldots,u_n)\in D^n$ with $F(u_1,\ldots,u_n)>0$, we have
\begin{equation}\label{introeq}
\sum_{x_2,\ldots,x_n\in D} F(u_1,x_2,\ldots,x_n)
  = s_1(u_1)\cdot \prod_{j=2}^n \frac{s_j(u_j)}{t_j(u_j)}.
\end{equation}
The intriguing part of (\ref{introeq}) is that its left side only depends on $u_1$
  but it holds for any $(u_1,\ldots,u_n)\in D^n$ as long as
$F(u_1,\ldots,u_n)>0$.
A crucial ingredient we use in constructing $t_2,\ldots,t_n$ and proving (\ref{introeq}) here is
  the succinct  data structure called {\it frame} introduced by Dyer and Richerby for
  unweighted \#CSP \cite{Dyer-Rich} (which is similar to the ``compact representation'' of Bulatov
  and Dalmau \cite{BulatovDalmau2}).
Once we have $t_2,\ldots,t_n$ and (\ref{introeq}), computing the partition function becomes trivial.

After obtaining the dichotomy,
  we also show in Section \ref{sec:dec} that the tractability criterion (that is, whether $\calF$ is balanced or not)
  is decidable in NP.
The proof follows the approach of Dyer and Richerby \cite{Dyer-Rich3} for unweighted \#CSP, with new
  ideas and constructions developed for the weighted setting.


This advance, from unweighted to weighted \#CSP, is akin to the leap from the
Dyer-Greenhill result on counting $0$-$1$ graph homomorphisms \cite{DyerGreenhill}
to the Bulatov-Grohe result for the non-negative case \cite{BulatovGrohe}.
The Bulatov-Grohe result paved the way for all future
developments.  This is because not only the Bulatov-Grohe
 result is intrinsically
important and sweeping
but also they gave
an elegant dichotomy criterion, which allows its easy application.
Almost all future results in this area use the Bulatov-Grohe criterion.
Here our result covers all non-negative counting CSP.
It achieves a similar leap from
the 0-1 case of Bulatov and Dyer-Richerby,
and in the meanwhile, simplifies the dichotomy criterion. Therefore
it is hoped that it will also be useful for future research.



In hindsight, perhaps one may re-evaluate the {\it Algebraic Approach}.
We now know that there is another {\it Alge\-braic Approach}, based
primarily on matrix algebra rather than (relational) universal algebra,
which gives us a more direct and complete dichotomy theorem for
\#CSPs. It is perhaps also a case where the proper generalization,
namely weighted \#CSP, leads to a simpler resolution of the problem
than  the original unweighted \#CSP.

Weighted
  \#CSP has many special cases that have been studied
intensively.
Graph homomorphisms can be considered as a special case of
weighted \#CSP where there is only one binary constraint function.
There has been great advances made on graph homomorphisms
  \cite{DyerGreenhill,BulatovGrohe, DirectedHomo,CaiChen}.
Our dichotomy theorem generalizes all previous
dichotomy theorems where the constraint functions are
non-negative.  Looking beyond non-negatively weighted
counting type problems, in graph homomorphisms~\cite{GGJT,CCL,Thurley}
great progress has already been made.
To extend that to \#CSPs with real or even complex weights will require significantly more effort (even for directed graph homomorphisms~\cite{CaiChen}).
For Boolean \#CSP with complex weights, a dichotomy was obtained~\cite{STOC09}.
Going beyond CSP type problems, holographic algorithms and reductions
are aimed precisely at these counting problems where cancelation
is the main feature. The work on Holant problems and their dichotomy
theorems are the beginning steps in that direction~\cite{STOC09,planar,CHL10}.\newpage

\section{Preliminaries}\label{preliminary}

We start with some definitions about non-negative matrices.

Let $\MM$ be a non-negative $m\times n$ matrix.
We say $\MM$ is \emph{rectangular} if one can permute its
  rows and columns separately, so that $\MM$ becomes a block-diagonal matrix.
More exactly, $\MM$ is rectangular if there exist
  $s$ pairwise disjoint and nonempty subsets of $[m]$, denoted by $A_1,\ldots,A_s$,
  and $s$ pairwise disjoint and nonempty subsets of $[n]$,
  denoted by $B_1,\ldots,B_s$, for some $s\ge 0$, such that for all $i\in [m]$ and $j\in [n]$,
$$
M(i,j)>0\ \ \Longleftrightarrow\ \ \text{$i\in A_k$ and $j\in B_k$ for some $k\in [s]$.}
$$

Now let $\MM$ be a non-negative and rectangular $m\times n$ matrix with
  $s$ blocks $A_1\times B_1,\ldots,A_s\times B_s$.
We say it is \emph{block-rank-$1$} if the $A_k\times B_k$ sub-matrix of
  $\MM$, for every $k\in [s]$, is of rank $1$.

The two lemmas below then follow directly from the definition of block-rank-1 matrices:

\begin{lemm}\label{trivial}
Let $\MM$ be a block-rank-$1$ matrix with $s\ge 1$ blocks:
  $A_1\times B_1,\ldots,A_s\times B_s$.
If $i^*\in A_k$ and $j^*\in B_k$ for some $k\in [s]$, then for any $i\in A_k$ we have
$$
\frac{\sum_{j\in B_k} M(i,j)}{\sum_{j\in B_k} M(i^*,j)}=\frac{M(i,j^*)}{M(i^*,j^*)}.\vspace{0.06cm}
$$
\end{lemm}

\begin{lemm}\label{trivial2}
If $\MM$ is a non-negative matrix but is not block-rank-$1$,
  then there exist two rows of $\MM$ that
  are neither linearly dependent nor orthogonal.
\end{lemm}

\subsection{Basic \#P-Hardness About Counting Graph Homomorphisms}

Every symmetric and non-negative $n\times n$ matrix $\AA$
  defines a graph homomorphism (or partition) function $Z_\AA(\cdot)$ as follows:
Given any undirected graph $G=(V,E)$, we have
$$
Z_\AA(G)\stackrel{\text{def}}{=}
  \sum_{\xi:V\rightarrow [n]}\hspace{0.036cm} \prod_{uv\in E} A\big(\xi(u),\xi(v)\big).
$$
We need the following important result of Bulatov and Grohe \cite{BulatovGrohe}
  to derive the hardness part of our dichotomy:

\begin{theo}\label{bulatovtheo}
Let $\AA$ be a symmetric and non-negative matrix with algebraic entries,
then the problem of computing $Z_\AA(\cdot)$ is in polynomial time if $\AA$
  is \emph{block-rank-$1$}; and is \#P-hard otherwise.
\end{theo}

\subsection{Weighted \#CSPs}

Let $D=\{1,2,\ldots,d\}$ be the domain set, where the size $d$ will be considered as a constant.
A \emph{weighted} constraint language $\calF$ over the domain $D$
  is a finite set of functions $\{f_1,\ldots,f_h\}$ in which
  $f_i:D^{r_i}\rightarrow \mathbb{R}$ is an
  $r_i$-ary function over $D$ for some $r_i\ge 1$.
The arity $r_i$ of $f_i$, $i\in [h]$, the number of functions $h$ in $\calF$, as well as the
  values of $f_i$, will all be
  considered as constants (except in Section \ref{sec:dec} where the decidability
  of the dichotomy is discussed).
In this paper, we only consider \emph{non-negative}{ weighted} constraint languages
  in which every $f_i$ maps $D^{r_i}$ to \emph{non-negative} and \emph{algebraic} numbers.\newpage

The pair $(D,\calF)$ defines the following problem which we simply
  denote by $(D,\calF)$:\vspace{0.06cm}
\begin{flushleft}
\begin{enumerate}
\item Let $\xx =(x_1,\ldots,x_n)\in D^n$ be a set of $n$ variables over $D$.
The input is then a collection $I$ of $m$ tuples $(f,i_1,\ldots,i_r)$ in
  which $f$ is an $r$-ary function in $\calF$ and $i_1,\ldots,i_r\in [n]$.
We call $n+m$ the size of $I$.\vspace{-0.06cm}

\item The input $I$ defines the following function $F_I$ over $\xx=(x_1,\ldots,x_n)\in D^n$:
$$
F_I(\xx) \stackrel{\text{def}}{=}
  \prod_{(f,i_1,\ldots,i_r)\in I} f(x_{i_1},\ldots,x_{i_r}),\ \ \ \ \ \text{for every $\xx\in D^n$}.\vspace{-0.05cm}
$$
And the output of the problem is the following sum:
$$
Z(I) \stackrel{\text{def}}{=} \sum_{\xx\in D^n} F_I(\xx).
$$
\end{enumerate}
\end{flushleft}

\subsection{Reduction from Unweighted to Weighted \#CSPs}

A special case is when every function in the language is boolean.
In this case, we can view each of the functions as a relation.
We use the following notation for this special case.

An \emph{unweighted} constraint language $\Gamma$ over the domain set $D$
  is a finite set of relations $\{\Theta_1,\ldots,\Theta_h\}$ in which every $\Theta_i$ is
  an $r_i$-ary relation over $D^{r_i}$ for some $r_i\ge 1$.
The language $\Gamma$ defines the following problem which we denote by $(D,\Gamma)$:\vspace{0.06cm}
\begin{flushleft}
\begin{enumerate}
\item Let $\xx =(x_1,\ldots,x_n)\in D^n$ be a set of $n$ variables over $D$.
The input is then a collection $I$ of $m$ tuples $(\Theta,i_1,\ldots,i_r)$ in
  which $\Theta$ is an $r$-ary relation in $\Gamma$ and $i_1,\ldots,i_r\in [n]$.
We call $n+m$ the size of $I$.

\item The input $I$ defines the following relation $R_I$ over $\xx=(x_1,\ldots,x_n)\in D^n$:\vspace{-0.06cm}
$$
\text{$\xx\in R_I\ \Longleftrightarrow$\ \ for every tuple $(\Theta,i_1,\ldots,i_r)\in I$,
  we have $(x_{i_1},\ldots,x_{i_r})\in \Theta$.\vspace{-0.13cm}}
$$
And the output of the problem is the number of $\xx\in D^n$ in the relation $R_I$.
\end{enumerate}
\end{flushleft}


For any non-negative weighted constraint language $\calF=\{f_1,\ldots,f_h\}$,
  it is natural to define its corresponding unweighted constraint language
  $\Gamma=\{\Theta_1,\ldots,\Theta_h\}$, where
  $\xx\in \Theta_i$ if and only if $f_i(\xx)>0$, {for all $i\in [h]$ and $\xx\in D^{r_i}$.}
In Section \ref{sec:htog}, we give a polynomial-time reduction from
  $(D,\Gamma)$ to $(D,\calF)$.\vspace{0.08cm}

\begin{lemm}\label{lem:simple}
Problem $(D,\Gamma)$ is polynomial-time reducible to $(D,\calF)$.\vspace{-0.04cm}
\end{lemm}
\begin{coro}\label{coro:trivial}
If $(D,\calF)$ is not \#P-hard, then neither is $(D,\Gamma)$.
\end{coro}

\subsection{Strong Rectangularity}

In the proof of the complexity dichotomy theorem for unweighted \#CSPs
  \cite{Bulatov,Dyer-Rich}, an important necessary
  condition
  for $(D,\Gamma)$ being not \#P-hard is \emph{strong rectangularity}:\vspace{0.03cm}\newpage

\begin{defi}[Strong Rectangularity]
We say $\Gamma$ is \emph{strongly rectangular} if for any input
  $I$ of $(D,\Gamma)$ \emph{(}which defines an $n$-ary relation $R_I$
  over $(x_1,\ldots,x_n)\in D^n$\emph{)} and for any
  integers $a,b:1\le a <b\le n$, the following $d^a\times d^{b-a}$ matrix $\MM$ is \emph{rectangular}:
the rows of $\MM$ are indexed by $\uu\in D^a$ and the columns are indexed
  by $\vv\in D^{b-a}$, and
$$
M(\uu,\vv)=\Big|\big\{\ww\in D^{n-b}:(\uu,\vv,\ww)\in R_I\big\}\Big|,\ \ \ \ \
\text{for all $\uu\in D^a$ and $\vv\in D^{b-a}$.}
$$
For the special case when $b=n$, we have $M(\uu,\vv)=1$ if $(\uu,\vv)\in R_I$ and
  $M(\uu,\vv)=0$ otherwise.\vspace{0.03cm}
\end{defi}

The following theorem can be found in \cite{Bulatov} and \cite{Dyer-Rich}:\vspace{0.03cm}

\begin{theo}\label{theo:rec}
If $\Gamma$ is not strongly rectangular, then $(D,\Gamma)$ is \#P-hard.\vspace{0.03cm}
\end{theo}

As a result, if $(D,\calF)$ is not \#P-hard, then $\Gamma$ must be
  strongly rectangular by Corollary \ref{coro:trivial} and Theorem \ref{theo:rec},
  where $\Gamma$ is the unweighted language that corresponds to $\calF$.
The strong rectangularity of $\Gamma$ then gives us the following
  algorithmic results from \cite{Dyer-Rich3}, using the succinct and efficiently
  computable data structure called frame.
They turn out to be very useful later in
  the study of the original weighted problem $(D,\calF)$.
We start with some notation.

Let $I$ be an input instance of $(D,\Gamma)$ which defines a relation $R$ over
  $n$ variables $\xx=(x_1,\ldots,x_n)$.\vspace{0.03cm}

\begin{defi}
For any $i\in [n]$, we use ${\sf pr}_i R \subseteq D$ to denote the projection of
  $R$ on the $i$th coordinate: $a\in {\sf pr}_i R$ if and only if there exist
  tuples $\uu\in D^{i-1}$ and $\vv\in D^{n-i}$ such that $(\uu,a,\vv)\in R$.

We define the following relation $\sim_i$ on ${\sf pr}_i R$: $a\sim_i b$ if
  there exist tuples $\uu\in D^{i-1}$ and $\vv_a,\vv_b\in D^{n-i}$ such that
  $(\uu,a,\vv_a)\in R$ and $(\uu,b,\vv_b)\in R$.\vspace{0.03cm}
\end{defi}

\begin{lemm}[\cite{Dyer-Rich3}]\label{usefulalg}
If $\Gamma$ is strongly rectangular then given any input $I$
  of $(D,\Gamma)$ which defines a relation $R$, we have
\begin{flushleft}
\begin{itemize}
\item[\emph{(A).}] For any $i\in [n]$, we can compute the set ${\sf pr}_i R$
  in polynomial time in the size of $I$.
Moreover, for every $a\in {\sf pr}_i R$, we can find a tuple
  $\uu\in R$ such that $u_i=a$ in polynomial time.\vspace{-0.08cm}
\item[\emph{(B).}] For any $i\in [n]$, the relation $\sim_i$ must be an
  \emph{equivalence} relation
  and can be computed in polynomial\\ time.
We will use $\mathcal{E}_{i,k}\subseteq D$, $k=1,2,\ldots,$ to denote
  the equivalent classes of $\sim_i$.\vspace{-0.08cm}
\item[\emph{(C).}] For any equivalence class $\mathcal{E}_{i,k}$, we can find, in polynomial time, a tuple
  $\uu^{[i,k]}\in D^{i-1}$ as well as a tuple $\vv^{[i,k,a]}\in D^{n-i}$
  for each element $a\in \mathcal{E}_{i,k}$ such that
  $(\uu^{[i,k]},a,\vv^{[i,k,a]})\in R$ for all $a\in \mathcal{E}_{i,k}.$\vspace{0.06cm}
\end{itemize}
\end{flushleft}
\end{lemm}

As a corollary, if $(D,\calF)$ is not \#P-hard, then we are able to use all
  the algorithmic results above for $(D,\Gamma)$ as subroutines, in the quest
  of finding a polynomial-time algorithm for $(D,\calF)$.

\section{A Dichotomy for Non-negative Weighted \#CSPs and its Decidability}

In this section, we prove a dichotomy theorem for all non-negative
  weighted \#CSPs and show that the characterization can be
  checked in NP.
The lemmas used in the proofs will be proved in the rest of the paper.

In the proof of our dichotomy theorem as well as its decidability, the following
  two notions of \emph{weak balance} and \emph{balance} play a crucial role.
It is similar to and, in some sense, weaker
  than the concept of \emph{strong balance} used in \cite{Dyer-Rich3}.
(Notably we do not use any existential quantifier in the definitions.)

\begin{defi}[Weak Balance]\label{def:sb}
We say $\calF$ is \emph{weakly balanced} if for any input instance $I$ of $(D,\calF)$
  \emph{(}which defines a non-negative function $F(x_1,\ldots,x_n)$ over $D$\emph{)}
  and for any integer $a:1\le a< n$, the following
  $d^a\times d$ matrix $\MM$ is \emph{block-rank-$1$}:
the rows of $\MM$ are indexed by $\uu\in D^{a}$ and the columns
  are indexed by $v\in D $, and
$$
M(\uu,v)= \sum_{\ww\in D^{n-a-1}} F(\uu,v,\ww),\ \ \ \ \ \text{for all
  $\uu\in D^{a}$ and $v\in D $.}
$$
For the special case when $a+1=n$, we have $M(\uu,v)=F(\uu,v)$ is block-rank-$1$.
\end{defi}

\begin{defi}[Balance]\label{defi:balance}
We call $\calF$ \emph{balanced} if for any input instance $I$ of $(D,\calF)$
  \emph{(}which defines a non-negative function $F(x_1,\ldots,x_n)$ over $D$\emph{)}
  and for any integers $a,b:1\le a<b\le n$, the following
  $d^a\times d^{b-a}$ matrix $\MM$ is \emph{block-rank-$1$}:
the rows of $\MM$ are indexed by $\uu\in D^{a}$ and the columns
  are indexed by $\vv\in D^{b-a} $, and
$$
M(\uu,\vv)= \sum_{\ww\in D^{n-b}} F(\uu,\vv,\ww),\ \ \ \ \ \text{for all
  $\uu\in D^{a}$ and $\vv\in D^{b-a} $.}
$$
For the special case when $b=n$, we have $M(\uu,\vv)=F(\uu,\vv)$ is block-rank-$1$.
\end{defi}

It is clear that \emph{balance} implies \emph{weak balance}.
We prove the following complexity dichotomy theorem.\vspace{0.06cm}

\begin{theo}\label{maintheorem}
$(D,\calF)$ is in P if $\Gamma$ is {strongly rectangular} and $\calF$ is
  {weakly balanced}; and is \#P-hard otherwise.\vspace{0.06cm}
\end{theo}
\begin{proof}
Assume $(D,\calF)$ is not \#P-hard.
By Corollary \ref{coro:trivial} and Theorem \ref{theo:rec},
  $\Gamma$ must be strongly rectangular. 
We prove the following lemma in Section \ref{app:hardness}, showing that $\calF$
  must be balanced and thus, weakly balanced:\vspace{0.06cm}

\begin{lemm}\label{lem:hard}
If $\calF$ is not balanced, then $(D,\calF)$ is \#P-hard.\vspace{0.06cm}
\end{lemm}

In the next two sections (Sections \ref{sec:vector} and \ref{sec:counting})
  we focus on the proof of the following algorithmic lemma:\vspace{0.06cm}

\begin{lemm}\label{lem:poly}
If $\Gamma$ is strongly rectangular and $\calF$ is weakly balanced,
  then $(D,\calF)$ is in polynomial time.\vspace{0.06cm}
\end{lemm}

The dichotomy theorem then follows directly.
\end{proof}

While the characterization of the dichotomy in Theorem \ref{maintheorem} above
  is very useful in the proof of its decidability, we can
  easily simplify it without using strong rectangularity.
We prove the following equivalent characterization
  using the notion of balance:

\begin{lemm}\label{lem:equiv}
$(D,\calF)$ is in polynomial time if $\calF$ is {balanced};
  and is \#P-hard otherwise.
\end{lemm}
\begin{proof}
Assume $(D,\calF)$ is not \#P-hard; otherwise we are already done.
By Lemma \ref{lem:hard}, we know $\calF$ must be balanced.
By Theorem \ref{maintheorem}, it suffices to show that if $\calF$ is balanced, then $\Gamma$
  is strongly rectangular, where we use $\Gamma$ to denote the unweighted
  constraint language that corresponds to $\calF$.
This follows directly from the definitions of strong rectangularity and
  balance, since a matrix that is block-rank-$1$ must first be rectangular.
\end{proof}

Next, we show that the complexity dichotomy is efficiently decidable.
Given $D$ and $\calF$,
  the decision problem of whether $(D,\calF)$ is in P or \#P-hard is
  actually in NP.
(Note that here $D$ and $\calF=\{f_1,\ldots,f_h\}$
  are considered no longer as constants, but as the input of the decision problem.
The input size is $d$ plus the number of bits needed to
  describe $f_1,\ldots,f_h$.)
We prove the following theorem in Section \ref{sec:dec}.
The proof follows the approach of Dyer and Richerby \cite{Dyer-Rich3}, with new
  ideas and constructions developed for the more general weighted case.
It uses a method of Lov\'asz \cite{Lovasz}, which was also used in \cite{DirectedHomo}.\vspace{0.03cm}\newpage

\begin{theo}\label{theo:dec}
Given $D$ and $\calF$,
  the problem of deciding whether $(D,\calF)$ is in P or \#P-hard is in \emph{NP}.
\end{theo}

\section{Vector Representation}\label{sec:vector}

Assume $\calF$ is weakly balanced, and let $f$ be an
  $r$-ary function in $\calF$.
We use $\Theta$ to denote the corresponding $r$-ary relation of $f$ in $\Gamma$.
In this section, we show that there must exist $r$ non-negative one-variable functions
$
s_1,\ldots,s_r:$ $D\rightarrow \mathbb{R}_+,
$
such that for all $\xx\in D^r$, either $\xx\notin \Theta$ and $f(\xx)=0$; or we have
$
f(\xx)= s_1(x_1)\cdots s_r(x_r).
$
We call any $\ss =$ $(s_1,\ldots, s_r)$ that satisfies the property
  above a \emph{vector representation} of $f$.
We prove the following lemma:\vspace{0.03cm}
\begin{lemm}\label{lem:vector}
If $\calF$ is weakly balanced, then every function $f\in \calF$ has a vector representation.\vspace{0.03cm}
\end{lemm}

To this end we need the following notation.
Let $f$ be any $r$-ary function over $D$.
Then for any $\ell\in [r]$, we use $f^{[\ell]}$ to denote the following $\ell$-ary function over $D$:
$$
f^{[\ell]}(x_1,\ldots,x_\ell)\hspace{0.06cm}\stackrel{\text{def}}{=}
\sum_{x_{\ell+1},\ldots,x_r\in D} f(x_1,\ldots,x_\ell,x_{\ell+1},\ldots,x_r),
\ \ \ \ \ \ \text{for all $x_1,\ldots,x_\ell\in D$.}
$$
In particular, we have $f^{[r]} \equiv f$.

Let $f$ be an $r$-ary non-negative function with $r\ge 1$.
We say $f$ is \emph{block-rank-$1$} if either
  $r=1$; or
the following $d^{r-1}\times d$ matrix
  $\MM$ is block-rank-$1$: the rows of $\MM$ are indexed by $\uu\in D^{r-1}$ and
  the columns are indexed by $v$ $\in D$, and
  $M(\uu,v)=f(\uu,v)$ {for all $\uu\in D^{r-1}$ and $v\in D$.}

By the definition of weak balance, Lemma \ref{lem:vector} is a direct
  corollary of the following lemma:\vspace{0.03cm}

\begin{lemm}
Let $f(x_1,\ldots,x_r)$ be an $r$-ary non-negative function.
If $f^{[\ell]}$ is block-rank-$1$ for all $\ell\in [r]$, then $f$
has a vector representation $\ss$.\vspace{0.03cm}
\end{lemm}
\begin{proof}
We prove the lemma by induction on $r$, the arity of $f$.

The base case when $r=1$ is trivial.
Now assume for induction that the claim is true for all $(r-1)$-ary
  non-negative functions, for some $r\ge 2$.
Let $f$ be an $r$-ary non-negative function such that $f^{[\ell]}$ is block-rank-$1$
  for all $\ell\in [r]$.
By definition, it is easy to see that
$$
\left(f^{[r-1]}\right)^{[\ell]} = f^{[\ell]},\ \ \ \ \ \text{for all $\ell\in [r-1]$.}
$$
As a result, if we denote $f^{[r-1]}$, an $(r-1)$-ary non-negative function, by $g$,
  then $g^{[\ell]}$ is block-rank-$1$ for every $\ell \in$ $ [r-1]$.
Therefore, by the inductive hypothesis, $g=f^{[r-1]}$ has
  a vector representation $(s_1,\ldots,s_{r-1})$.

Finally, we show how to construct $s_r$ so that $(s_1,\ldots,s_{r-1},s_r)$
  is a vector representation of $f$.
To this end, we let $\MM$ denote the following $d^{r-1}\times d$ matrix:
  The rows are indexed by $\uu\in D^{r-1}$ and the columns are indexed by $v\in D$,
  and $M(\uu,v)=f(\uu,v)$ for all $\uu\in D^{r-1}$ and $v\in D$.
By the assumption we know that $\MM$ is block-rank-$1$.
Therefore, by definition, there exist pairwise disjoint and nonempty subsets of $D^{r-1}$, denoted by
  $A_1,\ldots,A_s$, and pairwise disjoint and nonempty subsets of $D$, denoted by $B_1,\ldots,B_s$,
  for some $s\ge 0$, such that
$M(\uu,v)>0$ if, and only if $\uu\in A_i$ and $v\in B_i$ for some $i\in [s]$; and
for every $i\in [s]$, the $A_i\times B_i$ sub-matrix of $\MM$ is of rank $1$.

We now construct $s_r:D\rightarrow \mathbb{R}_+$ as follows.
For every $i\in [s]$, we arbitrarily pick a vector from $A_i$ and denote it $\uu_i$.
Then for $v\in D$, we set $s_r(v)$ as follows: \vspace{0.06cm}
\begin{enumerate}
\item If $v\notin B_i$ for any $i\in [s]$, then $s_r(v)=0$; and\vspace{-0.15cm}
\item Otherwise, assume $v\in B_i$. Then
\begin{equation}\label{eqpp}
s_r(v)=\frac{M(\uu_i,v)}{\sum_{v'\in B_i} M(\uu_i,v')}.
\end{equation}
\end{enumerate}

To prove that $(s_1,\ldots,s_r)$ is actually a vector representation of $f$, we only need
  to show that for every tuple $(\uu,v)$ such
  that $\uu\in A_i$ and $v\in B_i$ for some $i\in [s]$ (since otherwise we have
  $f(\uu,v)=0$), we have
$$
f(\uu,v)=M(\uu,v)=s_r(v)\prod_{j\in [r-1]} s_j(u_j).
$$
By using Lemma \ref{trivial} and (\ref{eqpp}), we have
$$
M(\uu,v)=M(\uu_i,v)\cdot \frac{\sum_{v'\in B_i}M(\uu,v')}{\sum_{v'\in B_i} M(\uu_i,v')}=
s_r(v)\cdot f^{[r-1]}(\uu)=s_r(v)\prod_{j\in [r-1]} s_j(u_j),
$$
where the last equation above follows from the inductive hypothesis
  that $(s_1,\ldots,s_{r-1})$ is a vector representation of $g=f^{[r-1]}$.
This finishes the induction, and the lemma is proved.
\end{proof}

\section{Tractability: The Counting Algorithm}\label{sec:counting}

In this section, we prove Lemma \ref{lem:poly} by giving a polynomial-time algorithm
  for the problem $(D,\calF)$, assuming $\Gamma$ is strongly rectangular and $\calF$ is weakly balanced.
As mentioned earlier, because $\Gamma$ is strongly rectangular
  we can use the three polynomial-time algorithms described
  in Lemma \ref{usefulalg} as subroutines.
Also because $\calF$ is weakly balanced, we
  may assume, by Lemma \ref{lem:vector}, that every $r$-ary function $f$ in $\calF$
  has a vector representation $\ss_f=(s_{f,1},\ldots,s_{f,r})$,
  where $s_{f,i}:D\rightarrow \mathbb{R}_+$ for all $i\in [r]$.

Now let $I$ be an input instance of $(D,\calF)$ and let $F$
  denote the function it defines over $\xx=(x_1,\ldots,x_n)\in D^n$.
For each tuple in $I$, one can replace the first component, that is, a function $f$ in $\calF$,
  by its corresponding relation $\Theta$ in $\Gamma$.
We use $I'$ to denote the
  new set, which is clearly an input instance of $(D,\Gamma)$ and defines
  a relation $R$ over $\xx\in D^n$.
We have $F(\xx)>0$ if and only if $\xx\in R$, for all $\xx\in D^n$.

The first step of our algorithm is to
  construct a vector representation $\ss=(s_1,\ldots,s_n)$ of $F$,
  using the vector representations $\ss_f$ of $f$, $f\in \calF$:\vspace{0.06cm}

\begin{lemm}
Given $I$, one can compute $s_1(\cdot),\ldots,s_n(\cdot)$ in polynomial time
  such that for all $\xx\in D^n$, either $\xx\notin R$ and $F(\xx)=0$; or
$F(\xx)=s_1(x_1)\cdots s_n(x_n).$ \vspace{0.06cm}
\end{lemm}
\begin{proof}
We start with $s_1,\ldots,s_n$ where $s_i(a)=1$ for all $i\in [n]$ and $a\in D$.
We then enumerate the tuples in $I$ one by one.
For each $(f,i_1,\ldots,i_r)\in I$ and each $j\in [r]$,
  we update the function $s_{i_j}(\cdot)$ using $s_{f,j}(\cdot)$ as follows:\vspace{-0.06cm}
$$
s_{i_j}(a) \stackrel{\text{set}}{=} s_{i_j}(a)\cdot s_{f,j}(a),\ \ \ \ \ \text{for every $a\in D$.}
$$
It is easy to check that the tuple $(s_1,\ldots,s_n)$ we get is a vector representation of $F$.
\end{proof}\newpage

The second step of the algorithm is to construct a sequence of one-variable functions
  $t_n(\cdot),t_{n-1}(\cdot),\ldots,t_2(\cdot)$ that
  have the following nice property:
for any $i\in \{1,\ldots,n-1\}$ and for any $\uu\in R$, we have \vspace{0.05cm}
\begin{equation}\label{imp}
\sum_{x_{i+1},\ldots,x_n\in D} F(u_1,\ldots,u_i,x_{i+1},\ldots,x_{n})=
s_1(u_1)\cdots s_i(u_i)\cdot \frac{s_{i+1}(u_{i+1})}{t_{i+1}(u_{i+1})} \cdots
  \frac{s_n(u_n)}{t_n(u_n)} .
\end{equation}
Before giving the construction and proving (\ref{imp}), we show that $Z(I)$ is easy to compute once
  we have $t_n,\ldots,t_2$.

For this purpose, we first compute ${\sf pr}_1 R$ in polynomial time using
  the algorithm in Lemma \ref{usefulalg} (A).
In addition, we find a vector $\uu_a=(u_{a,1},u_{a,2},\ldots,u_{a,n})\in R$
  for each $a\in {\sf pr}_1 R$ such that $u_{a,1}=a$ in polynomial time.
Then\vspace{0.03cm}
\begin{equation*}
Z(I)=\sum_{\xx\in D^n} F(\xx)
=\sum_{a\in {\sf pr}_1 R}\hspace{0.1cm} \sum_{x_2,\ldots,x_n\in D} F(a,x_2,\ldots,x_n)
=\sum_{a\in {\sf pr}_1 R}\hspace{0.1cm} s_1(a)\prod_{j\in [2:n]}
\left(\frac{s_j(u_{a,j})}{t_j(u_{a,j})}\right),\vspace{0.07cm}
\end{equation*}
which clearly can be evaluated in polynomial time using $s_1,\ldots,s_n$ and $t_2,\ldots,t_n$.

Now we construct $t_n,t_{n-1},\ldots,t_2$ and prove (\ref{imp})
  by induction. We start with $t_n(\cdot)$.

Because $\calF$ is weakly balanced, the following $d^{n-1}\times d$ matrix $\MM$ must be block-rank-$1$:
the rows are indexed by $\uu \in D^{n-1}$ and
  the columns are indexed by $v\in D$, and $M(\uu,v)=F(\uu,v)$
  for all $\uu\in D^{n-1}$ and $v\in D$.
By the definition of $\sim_n$, we have $v_1\sim_n v_2$ if and only if
  columns $v_1$ and $v_2$ are in the same block of $\MM$ and thus,
  the equivalent classes $\{\mathcal{E}_{n,k}\}$ are exactly the column index sets
  of those blocks of $\MM$.

We define $t_n(\cdot)$ as follows. For every $a\in D$, if $a\notin {\sf pr}_n R$ then $t_n(a)=0$;
Otherwise, $a$ belongs to one of the equivalence classes $\mathcal{E}_{n,k}$
  of $\sim_n$ and\vspace{0.03cm}
\begin{equation}\label{needed}
t_n(a)=\frac{s_n(a)}{\sum_{b\in \mathcal{E}_{n,k}} s_n(b)}.\vspace{0.07cm}
\end{equation}
By using the algorithm in Lemma \ref{usefulalg} (B) $t_n(\cdot)$ can
  be constructed efficiently.
We now prove (\ref{imp}) for $i=n-1$.
Given any $\uu\in R$, we have $u_n\in {\sf pr}_n R$ by definition and let
  $\mathcal{E}_{n,k}$ denote the equivalence class that $u_n$ belongs to.
Then\vspace{0.08cm}
$$
\sum_{b\in D} F(u_1,\ldots,u_{n-1},b)=
\sum_{b\in \mathcal{E}_{n,k}} F(u_1,\ldots,u_{n-1},b)
=\prod_{j\in [n-1]} s_j(u_j) \sum_{b\in \mathcal{E}_{n,k}} s_n(b)
=\prod_{j\in [n-1]} s_j(u_j) \cdot \frac{s_n(u_n)}{t_n(u_n)}.\vspace{0.12cm}
$$
The last equation follows from the construction
  (\ref{needed}) of $t_n(\cdot)$ and the assumption that $u_n\in \mathcal{E}_{n,k}$.

Now assume for induction that we already constructed $t_{i+1},\ldots,t_n$, for some
  $i\in [2:n-1]$, and they satisfy (\ref{imp}).
To construct $t_{i}(\cdot)$, we first observe that the following
  $d^{i-1}\times d$ matrix $\MM$ must be
  block-rank-$1$, because $\calF$ is weakly balanced:
the rows are indexed by $\uu=(u_1,\ldots,u_{i-1})\in D^{i-1}$ and the columns are
  indexed by $v\in D$,
$$
M(\uu,v) = \sum_{\ww\in D^{n-i}} F(\uu,v,\ww)=\sum_{(\uu,v,\ww)\in R} F(\uu,v,\ww).
$$
Similarly, by the definition of $\sim_i$, its equivalent classes $\{\mathcal{E}_{i,k}\}$
  are precisely the column index sets of those blocks of $\MM$.
By (\ref{imp}) and the inductive hypothesis we immediately have the following concise form for
  $M(\uu,v)$: for any $\ww=(w_{i+1},\ldots,w_n)\in D^{n-i}$ such that $(\uu,v,\ww)\in R$,
  we have\vspace{0.01cm}
\begin{equation}\label{concise}
M(\uu,v)=\left(\prod_{j\in [i-1]} s_j(u_j)\right)s_i(v)\left(
\prod_{j\in [i+1:n]}\frac{s_{j}(w_j)}{t_j(w_j)}\right).\vspace{0.12cm}
\end{equation}
Note that by (\ref{imp}), the choice of $\ww$ can be arbitrary as long as
  $(\uu,v,\ww)\in R$.

We now construct $t_i(\cdot)$. For every $a\in D$,\vspace{0.06cm}
\begin{flushleft}
\begin{enumerate}
\item If $a\notin {\sf pr}_i R$, then $t_i(a)=0$; and\vspace{-0.06cm}
\item Otherwise, let $\mathcal{E}_{i,k}$ denote the equivalence class
  of $\sim_i$ that $a$ belongs to.
Then by using the algorithm in Lemma \ref{usefulalg} (C), we find
  a tuple $\uu^{[i,k]}\in D^{i-1}$ and a tuple $\vv^{[i,k,b]}\in D^{n-i}$ for
  each $b\in \mathcal{E}_{i,k}$
  such that
$$
\big(\uu^{[i,k]},b,\vv^{[i,k,b]}\big)\in R,\ \ \ \ \ \text{for all $b\in \mathcal{E}_{i,k}$.}
$$
Then we set\vspace{0.12cm}
\begin{equation}\label{lalala}
t_i(a)=\frac{M(\uu^{[i,k]},a)}{\sum_{b\in \mathcal{E}_{i,k}} M(\uu^{[i,k]},b)}.\vspace{0.15cm}
\end{equation}
By (\ref{concise}), $t_i(a)$ can be computed efficiently using tuples $\uu^{[i,k]}$ and $\vv^{[i,k,b]}$,
  for $b\in \mathcal{E}_{i,k}$.\vspace{0.06cm}
\end{enumerate}
\end{flushleft}
This finishes the construction of $t_i(\cdot)$.

Finally we prove (\ref{imp}).
Let $\uu$ be any tuple in $R$ and $\mathcal{E}_{i,k}$ be
  the equivalence class of $\sim_i$ that $u_i$ belongs to. Then \vspace{0.06cm}
$$
\sum_{x_{i},\ldots,x_n\in D} F(u_1,\ldots,u_{i-1},x_i,\ldots,x_n)
=\sum_{b\in \mathcal{E}_{i,k}} \hspace{0.06cm}
\sum_{x_{i+1} ,\ldots,x_n\in D} F(u_1,\ldots,u_{i-1},b,x_{i+1},\ldots,x_n).
$$
Let $\uu^*$ denote the $(i-1)$-tuple $(u_1,\ldots,u_{i-1})$.
Then by the definition of $\MM$, we can rewrite the sum as\vspace{0.03cm}
$$
\sum_{x_{i},\ldots,x_n\in D} F(u_1,\ldots,u_{i-1},x_i,\ldots,x_n)
=\sum_{b\in \mathcal{E}_{i,k}} M(\uu^*,b).
$$
Recall the tuples $\uu^{[i,k]}$ and $\vv^{[i,k,b]}$, $b\in \mathcal{E}_{i,k}$,
  which we used in the construction of $t_i(\cdot)$.
Because $\MM$ is block-rank-$1$ and because $\uu^*$ and $\uu^{[i,k]}$
  are known to belong to the same block of $\MM$, we have\vspace{0.06cm}
$$
\sum_{b\in \mathcal{E}_{i,k}} M(\uu^*,b)
=\sum_{b\in \mathcal{E}_{i,k}} \frac{M(\uu^*,u_i)}{M(\uu^{[i,k]},u_i)}\cdot M(\uu^{[i,k]},b)
=\frac{M(\uu^*,u_i)}{M(\uu^{[i,k]},u_i)}\cdot\sum_{b\in \mathcal{E}_{i,k}} M(\uu^{[i,k]},b).
$$
However, by the definition (\ref{lalala}) of $t_i(\cdot)$, we have
$$
\sum_{b\in \mathcal{E}_{i,k}} M(\uu^{[i,k]},b)= \frac{M(\uu^{[i,k]},u_i)}{t_i(u_i)},
$$
since we assumed that $u_i\in \mathcal{E}_{i,k}$.
As a result, we have\vspace{0.06cm}
$$
\sum_{x_i,\ldots,x_n\in D} F(u_1,\ldots,u_{i-1},x_i,\ldots,x_n)
=\sum_{b\in \mathcal{E}_{i,k}} M(\uu^*,b)
=\frac{M(\uu^*,u_i)}{t_i(u_i)}
=\left(\prod_{j\in [i-1]} s_j(u_j)\right)
  \left(\prod_{j\in [i:n]} \frac{s_j(u_j)}{t_j(u_j)}\right).\vspace{0.18cm}
$$
The last equation follows from (\ref{concise}).
This finishes the construction of $t_n,\ldots,t_2$ and the proof of Lemma \ref{lem:poly}.

\section{Decidability of the Dichotomy}\label{sec:dec}

In this section, we prove Theorem \ref{theo:dec} by showing that
  the decision problem is in NP.

By Theorem \ref{maintheorem} we need to decide, given $D$ and $\calF$,
  whether $\Gamma$ is strongly rectangular and $\calF$ is weakly balanced or not.
The first part can be done in NP \cite{Bulatov,Dyer-Rich3} by exhaustively searching
  for a Mal'tsev polymorphism.

\begin{lemm}[\cite{Bulatov,Dyer-Rich3}]\label{lemm:gammanp}
Given $\Gamma$, deciding whether it is strongly rectangular is in \emph{NP}.
\end{lemm}

\subsection{Primitive Balance}

Next we show the notion of weak balance is
  equivalent to the following even \emph{weaker} notion of \emph{primitive balance}:

\begin{defi}[Primitive Balance]
We say $\calF$ is \emph{primitively balanced} if for any instance $I$ of
  $(D,\calF)$ and the $n$-ary function $F_I(x_1,\ldots,x_n)$ it defines,
  the following $d\times d$ matrix $\MM_I$ is block-rank-$1$: The rows
  of $\MM_I$ are indexed by $x_1\in D$ and the columns are indexed by $x_2\in D$, and
\begin{equation}\label{eq:blockrank}
M_I(x_1,x_2)=\sum_{x_3,\ldots,x_n\in D} F_I(x_1,x_2,x_3,\ldots,x_n),
\ \ \ \ \ \text{for all $x_1,x_2\in D$.}
\end{equation}
\end{defi}

It is clear that weak balance implies primitive balance. The following lemma proves the inverse direction:

\begin{lemm}
If $\Gamma$ is primitively balanced , then it is also weakly balanced.
\end{lemm}
\begin{proof}
Assume for a contradiction that $\Gamma$ is not weakly balanced.
By definition, this means there exist an $I$ over $n$-variables
  and an integer $a:1\le a<n$ such that the following $d^a\times d$ matrix
  $\MM$ is not block-rank-$1$: the rows of $\MM$ are indexed by $\uu\in D^a$ and
  the columns are indexed by $v\in D$, and
$$
M(\uu,v)=\sum_{\ww\in D^{n-a-1}} F_I(\uu,v,\ww),\ \ \ \ \ \text{for all $\uu\in D^a$ and $v\in D$.}
$$
As a result, we know by Lemma \ref{trivial2} that $\AA=\MM^{\text{T}}\MM$ is not
  block-rank-$1$.

To reach a contradiction, we construct $I'$ from $I$ as follows:
$I'$ has $2n-a$ variables in the following order:
$$
x_1,x_2,y_1,\ldots,y_a,z_1,\ldots, z_{n-a-1},w_1,\ldots,w_{n-a-1}.
$$
The instance $I'$ consists of two parts: a copy of $I$ over $(y_1,\ldots,y_a,x_1,z_1,\ldots,z_{n-a-1})$
  and a copy of $I$ over~$(y_1,\ldots,$ $y_a, x_2,w_1,\ldots,w_{n-a-1})$.
Let $F_{I'}$ denote the function that $I'$ defines.
It gives us the
  following $d\times d$ matrix $\MM_{I'}$:\vspace{0.06cm}
$$
M_{I'}(x_1,x_2)=\sum_{\yy\in D^a,\zz,\ww\in D^{n-a-1}}
  F_I(\yy,x_1,\zz)\cdot F_I(\yy,x_2,\ww)
=\sum_{\yy\in D^a} M(\yy,x_1)\cdot M(\yy,x_2)=A(x_1,x_2),
$$
which we know is not block-rank-$1$.
This contradicts with the assumption that $\calF$ is primitively balanced .
\end{proof}

Now the decision problem reduces to the following, and we call it
\begin{quote}
\textsc{primitive balance}: Given $D$ and $\calF$ such that
  $\Gamma$ is strongly rectangular (which\\ by Lemma \ref{lemm:gammanp} can be verified
  in NP), decide whether $\calF$ is primitively balanced.
\end{quote}\newpage

\noindent Since $\Gamma$ is strongly rectangular, we know that for any input $I$ of $(D,\calF)$,
  the $d\times d$ matrix $\MM_I$ defined in (\ref{eq:blockrank}) must be rectangular.
We need the following useful lemma from \cite{Dyer-Rich3}, which gives us
  a simple way to check whether a rectangular matrix is block-rank-$1$ or not.

\begin{lemm}[\cite{Dyer-Rich3}]\label{lemrank1}
A rectangular $d\times d$ matrix $\MM$ is block-rank-$1$ \emph{if and only if}
\begin{equation}\label{equivcond}
M({\alpha,\kappa})^2 M({\beta,\lambda})^2 M({\alpha,\lambda}) M({\beta,\kappa}) =
M({\alpha,\lambda})^2 M({\beta,\kappa})^2 M({\alpha,\kappa}) M({\beta,\lambda})
\end{equation}
{for all $\alpha\ne \beta\in D$ and $\kappa\ne \lambda\in D$.}
\end{lemm}
As a result, for \textsc{primitive balance} it suffices to check whether
  (\ref{equivcond}) holds for $\MM_I$, for all instances $I$ and for all $\alpha\ne \beta,\kappa\ne \lambda\in D$.
In the rest of this section, we fix $\alpha\ne \beta\in D$ and $\kappa\ne \lambda\in D$,
  and show that the decision problem (that is, whether (\ref{equivcond})
  holds for all $I$) is in NP.
Theorem \ref{theo:dec} then follows immediately
  since there are only polynomially many possible tuples
  $(\alpha,\beta,\kappa,\lambda)$ to check.

\subsection{Reformulation of the Decision Problem}

Fixing $\alpha\ne \beta\in D$ and $\kappa\ne \lambda\in D$,
  we follow \cite{Dyer-Rich3} and reformulate the decision problem
  using a new pair $(\frak D,\frak F)$, that is, the \emph{6-th power} of $(D,\calF)$:\vspace{0.1cm}
\begin{enumerate}
\item First, the new domain $\frak D=D^6$, and we use $\frak s=(s_1,\ldots,s_6)$
  to denote an element in $\frak D$, where $s_i\in D$.\vspace{-0.12cm}

\item Second, $\frak F=\{g_1,\ldots,g_h\}$ has the same number of functions as $\calF$ and
  every $g_i$, $i\in [h]$, has the same\\ arity $r_i$ as $f_i$.
Function $g_i:\frak D^{r_i}\rightarrow \mathbb{R}_+$ is constructed explicitly from $f_i$ as follows:
$$
g_i(\frak s_1,\ldots,\frak s_{r_i})=
\prod_{j\in [6]} f_i(s_{1,j},\ldots, s_{r_i,j}),\ \ \ \ \ \text{for all $\frak s_1,\ldots,
  \frak s_{r_i}\in \frak D=D^6$.}\vspace{-0.1cm}
$$
\end{enumerate}
In the rest of the section, we will always use $x_i$ to denote variables over $D$
  and $y_i,z_i$ to denote variables over $\frak D$.

Given any input instance $I$ of $(D,\calF)$ over $n$ variables $(x_1,\ldots,x_n)$,
  it naturally defines an
  input instance $\frak I$ of $(\frak D,\frak F)$ over $n$ variables $(y_1,\ldots,y_n)$ as
  follows: for each tuple $(f,i_1,\ldots,i_r)\in I$,
  add a tuple $(g, i_1,\ldots,i_r)$ to $\frak I$, where $g\in \frak F$
  corresponds to $f\in \calF$.
Moreover, this is clearly a bijection between
  the set of all $I$ and the set of all $\frak I$.
Similarly, we let $G:\frak D^n\rightarrow \mathbb{R}_+$ denote the $n$-ary
  function that $\frak I$ defines:
$$
G(y_1,\ldots,y_n)=\prod_{(g, i_1,\ldots,i_r)\in \frak I}
  g(y_{i_1},\ldots,y_{i_r}),
  \ \ \ \ \ \text{for all $y_1,\ldots,y_n\in \frak D.$}
$$
The reason why we introduce the new tuple $(\frak D,\frak F)$ is because
  it gives us a new and much simpler formulation of the decision problem we are interested.

To see this, we let $\frak a,\frak b,\frak c$ denote the following three specific
  elements from $\frak D$:
$$
\frak a=(\alpha,\alpha,\alpha,\beta,\beta,\beta),\ \ \ \
\frak b=(\kappa,\kappa,\lambda,\lambda,\lambda,\kappa),\ \ \ \
\frak c=(\lambda,\lambda,\kappa,\kappa,\kappa,\lambda).
$$
Since $\alpha\ne \beta$ and $\kappa\ne \lambda$, $\frak a,\frak b,\frak c$ are
  three distinct elements in $\frak D$.
We adopt the notation of \cite{Dyer-Rich3}. For each $\frak s\in \frak D$, let
$$
\text{hom}_{\frak s}(\frak I)\stackrel{\text{def}}{=}\sum_{y_3,\ldots,y_n\in \frak D} G(\frak a,\frak s,y_3,\ldots,y_n),
\ \ \ \ \ \text{for every instance $\frak I$ of $(\frak D,\frak F)$.}
$$
%
It is easy to prove the following two equations.
Let $\frak I$ be the instance of $(\frak D,\frak F)$ that corresponds to $I$,
  and $\MM_I$ be the $d\times d$ matrix as defined in (\ref{eq:blockrank}).
Then
\begin{eqnarray*}
&\text{hom}_{\frak b}(\frak I)=M_I({\alpha,\kappa})^2 M_I({\beta,\lambda})^2
  M_I({\alpha,\lambda}) M_I({\beta,\kappa}) &\ \text{and}\\[0.3ex]
&\text{hom}_{\frak c}(\frak I)=M_I({\alpha,\lambda})^2 M_I({\beta,\kappa})^2
  M_I({\alpha,\kappa}) M_I({\beta,\lambda})&
\end{eqnarray*}
As a result, we have the following reformulation of the decision problem:
$$
\text{$\MM_I$ satisfies (\ref{equivcond}) for all $I$\ \ \ $\Longleftrightarrow$\ \ \
  $\text{hom}_\frak b(\frak I)=\text{hom}_\frak c(\frak I)$ for all $\frak I$}
$$

The next reformulation considers sums over \emph{injective} tuples only.
We say $(y_1,\ldots,y_n)\in \frak D^n$ is an injective tuple if $y_i\ne y_j$ for all
  $i\ne j\in [n]$ (or equivalently, if we view $(y_1,\ldots,y_n)$ as a map from $[n]$ to $\frak D$,
  it is injective).
We use $Y_n$ to denote the set of injective $n$-tuples.
(Clearly this definition is only useful when $n\le |\frak D|$, otherwise $Y_n$ is empty.)
We now define functions $\text{mon}_{\frak s}(\frak I)$, which are sums over
  injective tuples: For each $\frak s\in \frak D$, let
$$
\text{mon}_{\frak s}(\frak I)\stackrel{\text{def}}{=}\sum_{(\frak a,\frak s,y_3,\ldots,y_n)\in Y_n}
G(\frak a,\frak s,y_3,\ldots,y_n),\ \ \ \ \ \text{for every instance $\frak I$ of $(\frak D,\frak F)$.}
$$

The following lemma shows that $\text{hom}_{\frak b}(\frak I)=
  \text{hom}_{\frak c}(\frak I)$ for all $\frak I$ if and only if the same equation
  holds for the sums over injective tuples.
The proof is exactly the same as Lemma 41 in \cite{Dyer-Rich3}, using the Mobius inversion.
So we skip it here.\vspace{0.05cm}

\begin{lemm}[\cite{Dyer-Rich3}, Lemma 41]\label{declem1}
$\text{\emph{hom}}_{\frak b}(\frak I)=
  \text{\emph{hom}}_{\frak c}(\frak I)$ for all $\frak I$ if and only if
$\text{\emph{mon}}_\frak b (\frak I) = \text{\emph{mon}}_\frak c(\frak I)$ for all $\frak I$.\vspace{0.05cm}
\end{lemm}

Finally, the following reformulation gives us a condition that can
  be checked in NP:\vspace{0.05cm}

\begin{lemm}\label{declem2}
$\text{\emph{mon}}_\frak b (\frak I) = \text{\emph{mon}}_\frak c(\frak I)$
  for all $\frak I$ if, and only if,
there exists a bijection $\pi$ from the domain $\frak D$ to itself \emph{(}which we will refer to
  as an \emph{automorphism} from $(\frak D,\frak F)$ to itself\emph{)} such that
$\pi(\frak a)=\pi(\frak a)$, $\pi(\frak b)=\pi(\frak c)$, and
  for every $r$-ary function $g\in \frak F$, we have
\begin{equation}\label{eq:equal}
  g(y_1,\ldots,y_r)= g\Big(\pi(y_1),\ldots,\pi(y_r)\Big),\ \ \ \ \ \
\text{for all $y_1,\ldots,y_r\in \frak D$.}
\end{equation}
\end{lemm}
\begin{proof}
We start with the easier direction:
If $\pi$ exists, then $\text{mon}_\frak b (\frak I) = \text{mon}_\frak c(\frak I)$
  for all $\frak I$.
This is because for any injective $n$-tuple $(\frak a,\frak b,y_3,\ldots,y_n)\in Y_n$,
  we can apply $\pi$ and get a new injective $n$-tuple $(\frak a,\frak c,\pi(y_3),\ldots,\pi(y_n))\in Y_n$
  and this is a bijection from
  $(\frak a,\frak b,y_3,\ldots,y_n)\in Y_n$ and $(\frak a,\frak c,z_3,\ldots,z_n)\in Y_n$.
Moreover, by (\ref{eq:equal}) we have
$$
G(\frak a,\frak b,y_3,\ldots,y_n)=G\big(\frak a,\frak c,\pi(y_3),\ldots,\pi(y_n)\big).
$$
As a result, the two sums $\text{mon}_\frak b(\frak I)$ and $\text{mon}_\frak c(\frak I)$
  over injective tuples must be equal.

The other direction is more difficult.
First, we prove that if $\text{mon}_\frak b (\frak I) = \text{mon}_\frak c(\frak I)$
  for all $\frak I$, then for any $\frak I$ and any tuple $(\frak a,\frak b,y_3,\ldots,y_n)\in Y_n$
  with $G(\frak a,\frak b,y_3,\ldots,y_n)>0$,
  there exists a $(\frak a,\frak c,z_3,\ldots,z_n)\in Y_n$ such that
\begin{equation}\label{lastequation}
G(\frak a,\frak b,y_3,\ldots,y_n)=G(\frak a,\frak c,z_3,\ldots,z_n).
\end{equation}
To prove this we look at the following sequence of instances $\frak J_1=\frak J,\frak J_2,\ldots$
  defined from $\frak I$, where $\frak J_j$ consists of exactly $j$ copies of $\frak J$ over the same set of variables.
We use $G_j$ to denote the $n$-ary function that $\frak J_j$ defines, then
$$
G_j(y_1,\ldots,y_n)=\big( G(y_1,\ldots,y_n)\big)^j,\ \ \ \ \ \text{for all
  $y_1,\ldots,y_n\in \frak D$.}
$$
Let $Q=\{q_1,\ldots,q_{|Q|}\}$ denote the set of all possible positive values of $G$ over $Y_n$;
  let $k_i\ge 0$ denote the number of tuples $(\frak a,\frak b,y_3,\ldots,y_n)\in Y_n$
  such that $G(\frak a,\frak b,y_3,\ldots,y_n)=q_i$, $i\in [|Q|]$;
  and let $\ell_i\ge 0$ denote the number of tuples $(\frak a,\frak c,y_3,\ldots,y_n)\in Y_n$
  such that $G(\frak a,\frak c,y_3,\ldots,y_n)=q_i$, $i\in [|Q|]$.
Then by $\text{mon}_\frak b(\frak I_j)=\text{mon}_\frak c(\frak I_j)$,\vspace{-0.02cm}
$$
\sum_{i\in [|Q|]} k_i\cdot (q_i)^j = \sum_{i\in [|Q|]} \ell_i\cdot (q_i)^j,
\ \ \ \ \ \text{for all $j\ge 1$.}
$$
Viewing $k_i-\ell_i$ as variables,
  the above equation gives us a linear system with a Vandermonde matrix
  if we let $j$ go from $1$ to $|Q|$.
As a result, we must have $k_i=\ell_i$ for all $i\in [|Q|]$,
  and (\ref{lastequation}) follows.

To finish the proof, we need the following technical lemma:\vspace{0.08cm}

\begin{lemm}\label{technical}
Let $Q$ be a finite and nonempty set of positive numbers.
Then for any $k\ge 1$, there exists a sequence of positive integers
  $N_1,\ldots,N_k$ such that
\begin{equation}\label{hahahaha}
q_1^{N_1} q_2^{N_2} \cdots q_k^{N_k}
=(q_1')^{N_1} (q_2')^{N_2} \cdots (q_k')^{N_k},\ \ \ \ \ \text{where $q_1,\ldots,q_k,q_1'\ldots,q_k'
\in Q$}
\end{equation}
if and only if $q_i=q_i'$ for every $i\in [k]$.\vspace{0.08cm}
\end{lemm}
\begin{proof}
The lemma is trivial if $|Q|=1$, so we assume $|Q|\ge 2$.
We use induction on $k$. The basis is trivial: we just set $N_1=1$.
Now assume the lemma holds for some $k\ge 1$, and $N_1,\ldots,N_k$ is the
  sequence for $k$.
We show how to find $N_{k+1}$ so that $N_1,\ldots,N_{k+1}$ satisfies the lemma for $k+1$.
To this end, we let
$$
c_\text{min}=\min_{q>q'\in Q} q/q'>1\ \ \ \ \ \ \text{and}\ \ \ \ \ \
c_\text{max}=\max_{q>q'\in Q} q/q'.
$$
Then we let $N_{k+1}$ be a large enough integer such that
$$
\big(c_\text{min}\big)^{N_{k+1}}> \big(c_\text{max}\big)^{\sum_{i\in [k]} N_i}.
$$

To prove the correctness, we assume (\ref{hahahaha}) holds.
First, we must have $q_{k+1}=q_{k+1}'$. Otherwise, assume without generality that
  $q_{k+1}>q_{k+1}'$, then by (\ref{hahahaha})
$$
\big(c_\text{min}\big)^{N_{k+1}}\le \big(q_{k+1}/q_{k+1}'\big)^{N_{k+1}}=\big(q_1'/q_1\big)^{N_1}\cdots
  \big(q_k'/q_k\big)^{N_k}\le \big(c_\text{max}\big)^{\sum_{i\in [k]}N_k},
$$
which contradicts with the definition of $N_{k+1}$.
Once we have $q_{k+1}=q_{k+1}'$, they can be removed from (\ref{hahahaha})
  and by the inductive hypothesis, we have $q_i=q_i'$ for all $i\in [k]$.
This finishes the induction, and the lemma is proved.\vspace{0.03cm}
\end{proof}

To find $\pi$, we define the following $\frak I$.
It has $|\frak D|$ variables and we denote them by $y_\frak s$, $\frak s\in \frak D$.
(In particular, $y_\frak a$ and $y_\frak b$ are the first and second variables of $\frak I$
  so that later $\text{mon}_\frak s(\frak I)$ is well-defined.)
Let $L$ be the set of all tuples $(g,\frak s_1,\ldots,\frak s_r)$,
  where $g$ is an $r$-ary function in $\frak F$ and
$
g(\frak s_1,\ldots,\frak s_r)>0.
$
We let $N_1,\ldots,N_{|L|}$ be the sequence of positive integers that satisfies
  Lemma \ref{technical} with $k=|L|$ and
$$
Q=\Big\{ g(\frak s_1,\ldots,\frak s_r): ( g,\frak s_1,\ldots,\frak s_r)\in L\Big\}.
$$
Then we enumerate all tuples in $L$ in any order. For the $i$th tuple
  $(g,\frak s_1,\ldots,\frak s_r)\in L$,
  $i\in [|L|]$, we add $N_i$ copies of
  the same tuple $(g,\frak s_1,\ldots,\frak s_r)$ to $\frak I$.
This finishes the definition of $\frak I$.

From the definition of $\frak I$, it is easy to see that
  $G(y_\frak s:y_\frak s=\frak s\ \text{for all $\frak s\in \frak D$})>0$.
Therefore, by (\ref{lastequation}) we know there exists a tuple $(z_\frak s:\frak s\in \frak D)\in Y_n$
  such that $z_\frak a =\frak a$, $z_\frak b=\frak c$, and
$$
G\big(y_\frak s:y_\frak s=\frak s\ \text{for all $\frak s \in \frak D$}\big)=
G\big(z_\frak s:\frak s\in \frak D\big)>0.
$$
We show that $\pi(\frak s)\stackrel{\text{def}}{=}
z_\frak s$, for every $\frak s\in \frak D$, is the bijection that we are looking for.

First, using Lemma \ref{technical}, it follows from the definition of $\frak I$ that
  for every tuple $(g,\frak s_1,\ldots,\frak s_r)\in L$, we have
$$
  g(\frak s_1,\ldots,\frak s_r)= g\big(\pi(\frak s_1),\ldots,\pi(\frak s_r)\big).
$$
So we only need to show that $ g(\pi(\frak s_1),\ldots,\pi(\frak s_r))=0$
  whenever $ g(\frak s_1,\ldots,\frak s_r)=0$.
This follows directly from the fact
  that $\pi$ is a bijection and thus,
  $(\frak s_1,\ldots,\frak s_r)\rightarrow (\pi(\frak s_1),\ldots,\pi(\frak s_r))$
  is also a bijection.\vspace{0.03cm}
\end{proof}

With Lemma \ref{declem1} and Lemma \ref{declem2}, we only need to check
  whether there exists an automorphism $\pi$ from $(\frak D,\frak F)$ to itself
  such that $\pi(\frak a)= \frak a $ and $\pi(\frak b)=\frak c$.
We can just exhaustively check all possible bijections from $\frak D$ to itself,
  and this gives us an algorithm in NP.

\section{Proof of Lemma \ref{lem:simple}}\label{sec:htog}

Let $I$ be an input of $(D,\Gamma)$ with $n$ variables $\xx=(x_1,\ldots,x_n)$ and $m$ tuples,
  and $R$ be the relation it defines.

For each $k\ge 1$, we let $I_k$ denote the following input of $(D,\calF)$:
$I_k$ has $n$ variables $(x_1,\ldots,x_n)$; and for each $(\Theta,i_1,\ldots,i_r)\in I$,
  we add $k$ copies of $(f,i_1,\ldots,i_r)$
  to $I_k$, where $f\in \calF$ is the $r$-ary function that corresponds to $\Theta\in \Gamma$.
We use $F_k(\xx)$ to denote the $n$-ary non-negative function that $I_k$ defines.
Then it is clear that
\begin{equation}\label{haha}
F_k(\xx)=\Big(F_1(\xx)\Big)^k,\ \ \ \ \ \text{for all $\xx\in D^n$.}
\end{equation}
We will show that to compute $|R|$, one only needs to evaluate $Z(I_k)$ for $k$ from
  $1$ to some polynomial of $m$.
This gives us a polynomial-time reduction from $(D,\Gamma)$ to $(D,\calF)$.

Now we let $Q_{m}$ denote the set of all integer tuples
$$
\qq=\Big(q_{i,\tt}\ge 0:\text{$i\in [h]$ and $\tt\in D^{r_i}$ such that $f_i(\tt)>0$}\Big)
$$
that sum to $m$. And let \textsc{Value}$_m$ denote
  the following set of positive numbers:
$$
\textsc{Value}_m=\left\{\prod_{i\in [h],\hspace{0.036cm}\tt\in D^{r_i}} \Big(f_i(\tt)\Big)^{q_{i,\tt}}:
  \qq\in Q_m\right\}.
$$
It is easy to show that both $|Q_m|$ and
  $|\textsc{Value}_m|$ are polynomial in $m$ (as $d,h$ and $r_i$, $i\in [h]$
  are all constants) and can be computed in polynomial time in $m$.
Moreover, by the definition of $\textsc{Value}_m$ we have for every $\xx\in D^n$:
$$
F_1(\xx)>0\ \ \Longrightarrow\ \ F_1(\xx)\in \textsc{Value}_m. 
$$

For every $c\in \textsc{Value}_m$, we let $N_c$ denote the number
  of $\xx\in D^n$ such that $F_1(\xx)=c$. Then we have
\begin{equation}\label{haha1}
Z(I_1)=\sum_{c\in \textsc{Value}_m} N_c\cdot c
\end{equation}
We also have
\begin{equation}\label{haha2}
\big|R\big|=\sum_{c\in \textsc{Value}_m} N_c
\end{equation}
and by (\ref{haha})
\begin{equation}\label{haha3}
Z(I_k)=\sum_{c\in \textsc{Value}_m} N_c\cdot c^k,\ \ \ \ \ \text{for every $k\ge 1$.}
\end{equation}
If we view $\{N_c:c\in \textsc{Value}_m \}$ as variables, then
  by taking $k=1,\ldots,|\textsc{Value}_m|$, (\ref{haha3})
  gives us a Vandermonde system from which we can compute $N_c$, $c\in \textsc{Value}_m$,
  in polynomial time.
We can then use (\ref{haha2}) to compute $|R|$.

This finishes the proof of Lemma \ref{lem:simple}.

\section{Proof of Lemma \ref{lem:hard}}\label{app:hardness}

Assume that $\calF$ is not balanced. Then by definition, there exists
  an input instance $I$ for $(D,\calF)$ such that
\begin{enumerate}
\item It defines an $n$-ary function $F(x_1,\ldots,x_n)$; and\vspace{-0.1cm}
\item There exist integers $a,b:1\le a<b\le n$ such that the following $d^a\times d^{b-a}$
  matrix $\MM$ is not\\ block-rank-$1$:
  the rows are indexed by $\uu\in D^a$ and the columns are indexed by $\vv\in D^{b-a}$,
  and
$$
M(\uu,\vv)=\sum_{\ww\in D^{n-b}} F(\uu,\vv,\ww),\ \ \ \ \ \text{for all $\uu\in D^{a}$
  and $\vv\in D^{b-a}$.}
$$
\end{enumerate}
Because $\MM$ is not block-rank-$1$, by Lemma \ref{trivial2},
  it has two rows that are neither linearly dependent nor orthogonal.
We let $\MM({\uu_1,*})$ and $\MM({\uu_2,*})$ be such two rows, where $\uu_1,\uu_2\in D^a$. Then
\begin{equation}\label{notrank}
0< \big\langle \MM({\uu_1,*}), \MM({\uu_2,*}) \big\rangle^2 <
\big\langle \MM({\uu_1,*}),\MM({\uu_1,*})\big\rangle \cdot \big\langle
\MM({\uu_2,*}),\MM({\uu_2,*})\big\rangle.
\end{equation}

We let $\AA=\MM\MM^{\text{T}}$, which is clearly a
  symmetric and non-negative $d^a\times d^a$ matrix, with both of its rows and columns
  indexed by $\uu\in D^a$.
It then immediately follows from (\ref{notrank}) that $\AA$ is not
  block-rank-$1$, since all the four entries in the $\{\uu_1,\uu_2\}\times
  \{\uu_1,\uu_2\}$ sub-matrix of $\AA$ are positive but this $2\times 2$
  sub-matrix is of rank $2$ by (\ref{notrank}).

To finish the proof, we give a polynomial-time reduction
  from $Z_\AA(\cdot)$ to $(D,\calF)$.
Because the former is \#P-hard by Theorem \ref{bulatovtheo} (since $\AA$ is not block-rank-$1$),
  we know that $(D,\calF)$ is also \#P-hard.

Let $G=(V,E)$ be an input undirected graph of $Z_\AA(\cdot)$.
We construct an input instance $I_G$ of $(D,\calF)$ from $G$, using $I$
  (which is considered
  as a constant here since it does not depend on $G$), as follows.\vspace{0.1cm}
\begin{enumerate}
\item For every vertex $v\in V$, we create $a$ variables over $D$, denoted by
  $x_{v,1},\ldots,x_{v,a}$; and\vspace{-0.12cm}
\item For every edge $e=vv'\in E$, we add $(b-a)+2(n-b)$ variables over $D$, denoted by\vspace{-0.015cm}
$$
y_{e,a+1},\ldots,y_{e,b},z_{e,b+1},\ldots,z_{e,n},z_{e,b+1}',\ldots,z_{e,n}'.\vspace{-0.015cm}
$$\newpage
Then we make a copy of $I$ over the following $n$ variables:\vspace{-0.015cm}
$$
\big(x_{v,1},\ldots,x_{v,a},y_{e,a+1},\ldots,y_{e,b},z_{e,b+1},\ldots,z_{e,n}\big)\vspace{-0.015cm}
$$
as well as the following $n$ variables:\vspace{-0.015cm}
$$
\big(x_{v',1},\ldots,x_{v',a},y_{e,a+1},\ldots,y_{e,b},z_{e,b+1}',\ldots,z_{e,n}'\big).\vspace{-0.015cm}
$$
\end{enumerate}
This finishes the construction of $I_G$.

It is easy to show by the definitions of $\MM$ and $\AA$ above
  that $Z_\AA(G)=Z(I_G)$.
This gives us a polynomial-time reduction from problems $Z_\AA(\cdot)$ to $(D,\calF)$ since
  $I_G$ can be constructed from $G$ in polynomial time.

\section{Equivalence of Balance and Strong Balance}\label{sec:equiv}

In \cite{Dyer-Rich} Dyer and Richerby used the following notion
  of strong balance for unweighted constraint languages $\Gamma$ and showed that
  $(D,\Gamma)$ is in polynomial time if $\Gamma$ is strongly balanced; and
  is \#P-hard otherwise.\vspace{0.06cm}

\begin{defi}
Let $\Gamma$ be an unweighted constraint language over $D$.
We call $\Gamma$ \emph{strongly balanced} if for every input instance $I$
  of $(D,\Gamma)$ \emph{(}which defines an $n$-ary relation $R$\emph{)}
  and for any $a,b,c:1\le a<b\le c\le n$, the following $d^a\times d^{b-a}$
  matrix $\MM$ is block-rank-$1$: the rows are indexed by $\uu\in D^a$
  and the columns are indexed by $\vv\in D^{b-a}$,
\begin{equation}\label{eq:yyyy}
M(\uu,\vv)= \Big|\big\{\ww\in D^{c-b}: \exists\hspace{0.05cm}\zz\in D^{n-c}\ \text{such that}\
  (\uu,\vv,\ww,\zz)\in R\big\}\Big|,\ \ \ \ \ \text{for all $\uu\in D^a$ and $\vv\in D^{b-a}$}.
\end{equation}
There are two special cases.
When $c=b$, $M(\uu,\vv)$ is $1$ if there exists a $\zz\in D^{n-c}$ such that $(\uu,\vv,\zz)\in R$;
  and is $0$ otherwise.
When $n=c$, $M(\uu,\vv)$ is the number of $\ww\in D^{c-b}$ such that $(\uu,\vv,\ww)\in R$.\vspace{0.06cm}
\end{defi}

\begin{theo}\label{theo:sb}
$(D,\Gamma)$ is in polynomial time if $\Gamma$ is strongly balanced; and is \#P-hard otherwise.\vspace{0.06cm}
\end{theo}

Notably the difference between the notion of  balance we used for weighted languages $\calF$
  (Definition~\ref{defi:balance})
  and the one above for unweighted languages $\Gamma$ \cite{Dyer-Rich} is that we do not
  allow the use of existential quantifiers in the former.
One can similarly define the following notion of balance for unweighted $\Gamma$:\vspace{0.06cm}

\begin{defi}
Let $\Gamma$ be an unweighted constraint language over $D$.
We call $\Gamma$ \emph{balanced} if for every instance $I$
  of $(D,\Gamma)$ \emph{(}which defines an $n$-ary relation $R$\emph{)}
  and for any $a,b :1\le a<b \le n$, the following $d^a\times d^{b-a}$
  matrix $\MM$ is block-rank-$1$: the rows are indexed by $\uu\in D^a$
  and the columns are indexed by $\vv\in D^{b-a}$,
\begin{equation}\label{eq:yyyy2}
M(\uu,\vv)= \Big|\big\{\ww\in D^{n-b}:
  (\uu,\vv,\ww )\in R\big\}\Big|,\ \ \ \ \ \text{for all $\uu\in D^a$ and $\vv\in D^{b-a}$}.\vspace{0.06cm}
\end{equation}
\end{defi}

We show below that these two notions, strong balance and balance, are equivalent.\vspace{0.06cm}



\begin{lemm}[Equivalence of Balance and Strong Balance]\label{lem:2}
If $\Gamma$ is balanced, then it is also strongly balanced.\vspace{0.06cm}
\end{lemm}

\begin{proof}
We assume that $\Gamma$ is balanced.
Let $I$ be any instance of $(D,\Gamma)$ which defines an $n$-ary relation $R$.
Let~$a,b$ and  $c$ be integers such that $1\le a<b\le c\le n$.
It suffices to show that the matrix $\MM$ in (\ref{eq:yyyy}) is block-rank-$1$.\newpage

For this purpose, we define a new input instance $I_k$ of $(D,\Gamma)$ for each $k\ge 1$:\vspace{0.06cm}
\begin{enumerate}
\item First, $I_k$ has $c+k(n-c)$ variables in the following order:
$$
x_1,\ldots,x_c,y_{1,c+1},\ldots,y_{1,n},\ldots,y_{k,c+1},\ldots,y_{k,n}.
$$
Below we let $\yy_i$, $i\in [k]$, denote $(y_{i,c+1},\ldots,y_{i,n})$ for convenience.\vspace{-0.1cm}

\item For each $i\in [k]$, we add a copy of $I$ on the following $n$ variables of $I_k$:
$
x_1,\ldots,x_c,y_{i,c+1},\ldots,y_{i,n}.
$
\end{enumerate}
It is clear that $I_1$ is exactly $I$.
We also use $R_k$ to denote the relation that $I_k$ defines, $k\ge 1$.

Because $\Gamma$ is balanced, the following $d^{a}\times d^{b-a}$ matrix $\MM^{[k]}$
  is block-rank-$1$: For $\uu\in D^a$ and $\vv\in D^{b-a}$,\vspace{0.03cm}
$$
M^{[k]}(\uu,\vv)= \left| \Big\{ (\ww,\yy_1,\ldots,\yy_k):
  \ww\in D^{c-b},\yy_1,\ldots,\yy_k\in D^{n-c}\ \text{and}\
  (\uu,\vv,\ww,\yy_1,\ldots,\yy_k)\in R_k\Big\} \right|.\vspace{0.03cm}
$$
From the definition of $I_k$, we have
$M(\uu,\vv)>0$ if and only if $M^{[k]}(\uu,\vv)>0$, for all
  $\uu\in D^a$ and $\vv\in D^{b-a}$.

Therefore, there exist pairwise disjoint and nonempty subsets of $D^a$,
  denoted $A_1,\ldots,A_s$, and pairwise disjoint and nonempty subsets of $D^{b-a}$,
  denoted $B_1,\ldots,B_s$, for some $s\ge 0$, such that
$$
M(\uu,\vv)>0\ \Longleftrightarrow\ M^{[k]}(\uu,\vv)>0\ \Longleftrightarrow\
\uu\in A_\ell\ \text{and}\ \vv\in B_\ell\ \text{for some $\ell\in [s]$.}
$$
Now to prove that $\MM$ is block-rank-$1$, we only need to show that for every $\ell\in [s]$,
\begin{equation}\label{finaleq}
M(\uu_1,\vv_1)\cdot M(\uu_2,\vv_2)=M(\uu_1,\vv_2)\cdot M(\uu_2,\vv_1),\ \ \ \ \
\text{for all $\uu_1,\uu_2\in A_\ell$ and $\vv_1,\vv_2\in B_\ell$.}
\end{equation}

To prove (\ref{finaleq}), we let
$$
W_{i,j}= \Big\{ \ww\in D^{c-b}: \exists\hspace{0.06cm}\yy\in D^{n-c}\ \text{such that}\
  (\uu_i,\vv_j,\ww,\yy)\in R\Big\},\ \ \ \ \ \text{for $i,j\in \{1,2\}$.}
$$
Furthermore, for every $\ww\in W_{i,j}$, we let
  $Y_{i,j,\ww}$ denote the (nonempty) set of $\yy\in D^{n-c}$ such that
  $(\uu_i,\vv_j,\ww,\yy)\in R$.
Now using $W_{i,j}$ and $Y_{i,j,\ww}$, it follows from the definition of $I_k$ that
$$
M^{[k]}(\uu_i,\vv_j)=\sum_{\ww\in W_{i,j}} \Big| Y_{i,j,\ww}\Big|^k.
$$
Because $\MM^{[k]}$ is block-rank-$1$, we have the following equation for every $k\ge 1$:
$$
\sum_{\ww\in W_{1,1},\ww'\in W_{2,2}} \Big(\hspace{0.08cm}\big| Y_{1,1,\ww}\big|\cdot
  \big| Y_{2,2,\ww'}\big|\hspace{0.08cm}\Big)^k =
\sum_{\ww\in W_{1,2},\ww'\in W_{2,1}} \Big(\hspace{0.08cm}\big| Y_{1,2,\ww}\big|\cdot
  \big| Y_{2,1,\ww'}\big|\hspace{0.08cm}\Big)^k.
$$
Since the equation above holds for every $k\ge 1$, the two sides must have the same number
  of positive terms.
By definition, we have $Y_{i,j,\ww}$ is nonempty for all $\ww\in W_{i,j}$.
As a result, we have
$$
|W_{1,1}|\cdot |W_{2,2}| = |W_{1,2}|\cdot |W_{2,1}|
$$
and (\ref{finaleq}) follows.
This finishes the proof of Lemma \ref{lem:2}.
\end{proof}

\newpage
\bibliographystyle{plain}
\begin{flushleft}
\bibliography{Reference}
\end{flushleft}

\end{document}